\documentclass[11pt,a4paper]{article}

\usepackage{jheppub}
\usepackage{amsthm,mathtools,bm}
\usepackage{booktabs,array}
\usepackage{microtype}

\newtheorem{theorem}{Theorem}[section]

\theoremstyle{remark}

\newcommand{\R}{\mathbb R}
\newcommand{\avg}[1]{\left\langle #1\right\rangle_t}

\newcommand{\Id}{\mathbf 1}
\newcommand{\p}{\partial}
\newcommand{\D}{\mathcal D}

\title{An exact computer-assisted lower bound for Fisher-information convexity along the two-dimensional heat flow}

\author{Fayez Abu-Ajamieh$^1$,}
\affiliation{$^1$Formerly: The Indian Institute of Science, Bangalore, India}
\emailAdd{fayezabuajamieh@gmail.com}

\author{Xiaoyan Huo$^2$,}
\affiliation{$^2$Universiti Teknologi Malaysia, Malaysia}
\emailAdd{huoxiaoyanphd@163.com}

\abstract{
We study the sharp constant in the differential inequality
\[
 I(t)I''(t)\geq c\,[I'(t)]^2
\]
for Fisher information along the probability-normalized heat flow on $\mathbb R^2$. The unrestricted log-convexity value $c=1$ is false in two dimensions, whereas the known square-root convexity theorem gives the universal value $c=1/2$. We construct an exact, machine-verifiable sum-of-squares certificate proving the stronger universal bound
\[
 c\geq \frac{231691}{250000}=0.926764.
\]
Together with an independent exact Fourier/Gaussian-envelope construction, this gives
\[
0.926764\leq c_*<0.94176665012062,
\]
where $c_*$ is the sharp constant for positive-time Gaussian convolutions of arbitrary probability measures on $\mathbb R^2$. We formulate the proof in explicit two-dimensional index notation. The lower bound follows from an exact rational polynomial identity involving a strictly positive denominator, a weighted total derivative, and eleven positive rational Gram matrices. All $491$ recorded leading principal minors are strictly positive and the complete polynomial identity has zero residual. No symmetry, log-concavity, finite-variance, or finite-entropy assumption is imposed on the initial measure. The sharp value remains open.
}

\begin{document}
\maketitle

\section{Introduction}

Let $X$ be an arbitrary random vector in $\R^2$ with probability law $\mu$, and let $Z\sim N(0,\Id)$ be an independent standard Gaussian. At positive time $t$ define
\begin{equation}
 Y_t=X+\sqrt t\,Z.
 \label{eq:Yt}
\end{equation}
The probability density of $Y_t$ is the Gaussian convolution
\begin{equation}
 f_t(x)=(\mu*\gamma_t)(x),\qquad
 \gamma_t(x)=\frac{1}{2\pi t}\exp\!\left(-\frac{|x|^2}{2t}\right).
 \label{eq:convolution}
\end{equation}
Because the Gaussian kernel itself satisfies
\[
 \p_t\gamma_t=\frac12\nabla^2\gamma_t,
\]
the convolution $f_t$ obeys the probability-normalized heat equation
\begin{equation}
 \boxed{\p_t f_t=\frac12\nabla^2 f_t.}
 \label{eq:heat}
\end{equation}
At every $t>0$, $f_t$ is smooth and strictly positive even when the initial measure $\mu$ is singular.

Define
\begin{equation}
 u(x,t)\equiv \log f_t(x).
 \label{eq:u}
\end{equation}
The Fisher information can then be written as
\begin{equation}
 I(t)=\int_{\R^2}d^2x\,f_t\,\p_i u\,\p_i u
 =\int_{\R^2}d^2x\,\frac{\p_i f_t\,\p_i f_t}{f_t}.
 \label{eq:Fisher}
\end{equation}
Repeated spatial indices $i,j,k,\ldots=1,2$ are summed. The question is how strongly the curvature of $I(t)$ is constrained along the heat flow. The strongest possible inequality of the form
\begin{equation}
 I I''\geq c(I')^2
 \label{eq:mainineq}
\end{equation}
is $c=1$, which is equivalent to log-convexity of $I$. Ledoux, Nair, and Wang proved the corresponding one-dimensional result \cite{LedouxNairWang}. In two dimensions, full log-convexity has been disproved by an explicit hexagonal construction \cite{ZouFanGaoWang}. Liu and Gao proved square-root convexity, corresponding to $c=1/2$ \cite{LiuGao}.

We define the sharp two-dimensional heat-flow constant by
\begin{equation}
 c_*\equiv\inf_{\mu,\,t>0}\frac{I(t)I''(t)}{[I'(t)]^2}.
 \label{eq:cstar}
\end{equation}
The recent counterexample work gives the qualitative interval $1/2\leq c_*<1$ \cite{ZouFanGaoWang}. Our main result is the much narrower rigorous bracket
\begin{equation}
 \boxed{\frac{231691}{250000}\leq c_*<0.94176665012062.}
 \label{eq:maininterval}
\end{equation}
The lower endpoint is not a floating-point claim. Numerical semidefinite programming is used only to discover a candidate sum-of-squares representation; the final certificate is reconstructed over the rationals and independently verified exactly.

\paragraph{Scope.}
The lower theorem is stated for the positive-time heat-flow class $f_t=\mu*\gamma_t$. We do not identify this class with a broader static density class under weaker integrability assumptions. The upper example is a smooth, strictly positive, Gaussian-decaying density and is transferred to positive heat time by continuity.

\section{Heat-flow identities in explicit index notation}
\label{sec:identities}

We use angle brackets for expectation with respect to $f_t$,
\begin{equation}
 \avg{\mathcal O}\equiv\int_{\R^2}d^2x\,f_t(x)\,\mathcal O(x).
 \label{eq:avg}
\end{equation}
Define the local derivative invariants
\begin{equation}
 r\equiv u_i u_i,
 \qquad
 q\equiv u_{ij}u_{ij},
 \qquad
 h\equiv u_{ii}=\nabla^2u,
 \label{eq:rqh}
\end{equation}
and
\begin{equation}
 d_3\equiv u_{ijk}u_{ijk}-2u_{ij}u_{jk}u_{ki}.
 \label{eq:d3}
\end{equation}
The subscript on $d_3$ refers to the appearance of third derivatives and is unrelated to the spatial dimension.

The three basic identities needed below are
\begin{equation}
 \boxed{I=\avg r,\qquad I'=-\avg q,\qquad I''=\avg{d_3}.}
 \label{eq:threeidentities}
\end{equation}
For completeness we give the derivation in the notation used throughout the paper.

Since $u=\log f_t$,
\begin{equation}
 \p_i f_t=f_t u_i.
 \label{eq:df}
\end{equation}
Thus, for any sufficiently regular scalar $\Phi$,
\begin{equation}
 \avg{\p_i\Phi}=-\avg{u_i\Phi}.
 \label{eq:weightedIBP}
\end{equation}
Equation \eqref{eq:Fisher} immediately gives $I=\avg r$. Integrating once by parts also gives
\begin{equation}
 I=-\avg h.
 \label{eq:IminusH}
\end{equation}

Next, using \eqref{eq:heat} and $f_t=e^u$,
\begin{equation}
 \p_tu=\frac{\p_t f_t}{f_t}
 =\frac12\frac{\nabla^2f_t}{f_t}
 =\frac12(h+r).
 \label{eq:ulogheat}
\end{equation}
Differentiating \eqref{eq:IminusH}, applying \eqref{eq:weightedIBP}, and using
\begin{equation}
 \nabla^2r=2q+2u_i\p_i h,
 \label{eq:lapr}
\end{equation}
one finds
\begin{equation}
 I'=-\frac12\avg{h^2+hr+2q+u_i\p_i h}.
 \label{eq:Iprimeintermediate}
\end{equation}
The integral of the total derivative $\p_i(f_t h u_i)$ vanishes, implying
\begin{equation}
 \avg{u_i\p_i h}=-\avg{h^2+hr}.
 \label{eq:mixid1}
\end{equation}
Substitution into \eqref{eq:Iprimeintermediate} therefore yields
\begin{equation}
 \boxed{I'=-\avg q.}
 \label{eq:Iprime}
\end{equation}

Differentiating once more,
\begin{equation}
 I''=-\avg{\p_tq+q\,\p_tu}.
 \label{eq:Isecondstart}
\end{equation}
Since $q=u_{ij}u_{ij}$,
\begin{equation}
 \p_tq=u_{ij}h_{ij}+2u_{ij}u_{jk}u_{ki}+2u_k u_{ij}u_{ijk}.
 \label{eq:qt}
\end{equation}
The first term is integrated by parts using $h_{ij}=\p_k u_{ijk}$,
\begin{equation}
 \avg{u_{ij}h_{ij}}
 =-\avg{u_k u_{ij}u_{ijk}}-\avg{u_{ijk}u_{ijk}}.
 \label{eq:hijIBP}
\end{equation}
Furthermore $\p_k q=2u_{ij}u_{ijk}$, and the vanishing integral of $\p_k(f_t q u_k)$ gives
\begin{equation}
 \avg{u_k u_{ij}u_{ijk}}+\frac12\avg{q(h+r)}=0.
 \label{eq:mixid2}
\end{equation}
All mixed terms then cancel, leaving
\begin{equation}
 \boxed{I''=\avg{u_{ijk}u_{ijk}-2u_{ij}u_{jk}u_{ki}}=\avg{d_3}.}
 \label{eq:Isecond}
\end{equation}

It is convenient to introduce
\begin{equation}
 Q\equiv\avg q=-I',\qquad D\equiv\avg{d_3}=I''.
 \label{eq:QD}
\end{equation}
Then the target inequality is simply
\begin{equation}
 ID\geq cQ^2.
 \label{eq:IDQ}
\end{equation}

\section{The exact lower-bound certificate}
\label{sec:certificate}

\subsection{Why an auxiliary parameter is useful}

The product structure in \eqref{eq:IDQ} is inconvenient for a local polynomial sum-of-squares argument, because $I$, $Q$, and $D$ are separate spatial averages. We therefore introduce a positive auxiliary parameter $\tau$ and consider
\begin{equation}
 F(\tau)\equiv D-2c\tau Q+c\tau^2 I
 =\avg{d_3-2c\tau q+c\tau^2r}.
 \label{eq:Ftau}
\end{equation}
If $F(\tau)\geq0$ for all $\tau>0$, then the quadratic is minimized at
\begin{equation}
 \tau_* = \frac QI,
 \label{eq:taustar}
\end{equation}
and hence
\begin{equation}
 F(\tau_*)=D-c\frac{Q^2}{I}\geq0
 \quad\Longrightarrow\quad
 ID\geq cQ^2.
 \label{eq:tauimplies}
\end{equation}
The role of $\tau$ is therefore to convert a nonlinear product of global averages into the expectation value of one local expression.

\subsection{A strictly positive denominator}

Set
\begin{equation}
 B\equiv r+\tau
 \label{eq:Bdef}
\end{equation}
and define
\begin{equation}
 H=B^2+q+\frac14Bh+\frac32u_i u_{ij}u_j.
 \label{eq:H}
\end{equation}
The last two terms are not manifestly positive, so we now show explicitly that $H$ is strictly positive.

Decompose the Hessian and the rank-one tensor $u_i u_j$ into trace and traceless pieces,
\begin{equation}
 u_{ij}=A^{(0)}_{ij}+\frac h2\delta_{ij},
 \qquad
 A^{(0)}_{ii}=0,
 \label{eq:A0}
\end{equation}
\begin{equation}
 u_i u_j=X_{ij}+\frac r2\delta_{ij},
 \qquad
 X_{ii}=0.
 \label{eq:Xdef}
\end{equation}
In two dimensions one finds
\begin{equation}
 A^{(0)}_{ij}A^{(0)}_{ij}=q-\frac{h^2}{2},
 \qquad
 X_{ij}X_{ij}=\frac{r^2}{2},
 \label{eq:A0Xnorms}
\end{equation}
and
\begin{equation}
 A^{(0)}_{ij}X_{ij}=u_i u_{ij}u_j-\frac{hr}{2}.
 \label{eq:A0Xcontract}
\end{equation}
Using these relations and completing squares gives the exact decomposition
\begin{equation}
\begin{aligned}
 H={}&
 \left(A^{(0)}_{ij}+\frac34X_{ij}\right)
 \left(A^{(0)}_{ij}+\frac34X_{ij}\right)
 +\frac12\left(h+r+\frac\tau4\right)^2
 \\
 &+\frac7{32}r^2+\frac74r\tau+\frac{31}{32}\tau^2.
\end{aligned}
 \label{eq:Hsquare}
\end{equation}
Every term on the right is nonnegative for $\tau>0$. Therefore
\begin{equation}
 \boxed{H\geq\frac{31}{32}\tau^2>0.}
 \label{eq:Hpositive}
\end{equation}
This is crucial because the certificate below contains $H^{-2}$ and $H^{-3}$.

\subsection{The local polynomial identity}

Let
\begin{equation}
 c_0=\frac{231691}{250000}.
 \label{eq:c0}
\end{equation}
For any vector field $V_i(x)$ define the weighted derivative
\begin{equation}
 \D_i V_i\equiv \p_iV_i+u_iV_i.
 \label{eq:Dweighted}
\end{equation}
Because $\p_i f_t=f_tu_i$,
\begin{equation}
 f_t\D_iV_i=\p_i(f_tV_i),
 \label{eq:weightedtotal}
\end{equation}
so $\D_iV_i$ is a total derivative with respect to the weighted measure $f_t d^2x$.

The exact certificate consists of a polynomial vector $J_i$, two sum-of-squares polynomials $S_0,S_1$, and the positive polynomial $H$ such that
\begin{equation}
 \boxed{
 H^3\left(d_3-2c_0\tau q+c_0\tau^2r\right)
 =S_0+\tau S_1+H\D_iJ_i-2J_i\p_iH.}
 \label{eq:certificate}
\end{equation}
Using the product rule,
\begin{equation}
 \D_i\!\left(\frac{J_i}{H^2}\right)
 =\frac{H\D_iJ_i-2J_i\p_iH}{H^3},
 \label{eq:weightedquotient}
\end{equation}
so after dividing by $H^3$,
\begin{equation}
 d_3-2c_0\tau q+c_0\tau^2r
 =\frac{S_0+\tau S_1}{H^3}
 +\D_i\!\left(\frac{J_i}{H^2}\right).
 \label{eq:divided}
\end{equation}
The first term on the right is pointwise nonnegative. The second will vanish after spatial integration.

\subsection{What the computer verifies}

At a fixed point, the independent derivative variables are
\begin{equation}
 u_1,u_2;
 \quad u_{11},u_{12},u_{22};
 \quad u_{111},u_{112},u_{122},u_{222};
 \quad u_{1111},u_{1112},u_{1122},u_{1222},u_{2222},
 \label{eq:jetvars}
\end{equation}
together with $\tau$. Thus \eqref{eq:certificate} is an exact polynomial identity in $15$ variables. Expanding both sides and matching every monomial gives $925$ exact linear coefficient equations. The fourth derivatives arise because $\p_iJ_i$ and $\p_iH$ differentiate the lower-order local variables; since the left-hand side contains no fourth derivatives, their exact cancellation is a stringent consistency check.

The nonnegative polynomials are represented by Gram matrices,
\begin{equation}
 S_0=\sum_{a=1}^{5}v_a^TM_av_a,
 \qquad
 S_1=\sum_{a=1}^{6}w_a^TN_aw_a.
 \label{eq:GramSOS}
\end{equation}
If a real symmetric matrix $M$ satisfies $M\succeq0$, then $v^TMv\geq0$ for every real vector $v$. The eleven rational Gram blocks have dimensions
\begin{equation}
 106,85,54,26,9,47,26,64,44,22,8.
 \label{eq:blocksizes}
\end{equation}
The exact verifier checks all $491$ recorded leading principal minors and finds them positive.

\subsection{Why the weighted total derivative vanishes}

The formal identity \eqref{eq:weightedtotal} is not by itself enough on the noncompact space $\R^2$; the boundary term must be controlled. For
\begin{equation}
 Y=X+\sqrt t\,Z,
 \end{equation}
Gaussian differentiation gives, for every multi-index $\alpha$,
\begin{equation}
 \frac{\p^\alpha f_t(Y)}{f_t(Y)}
 =(-1)^{|\alpha|}t^{-|\alpha|/2}\,
 \mathbb E[H_\alpha(Z)\mid Y],
 \label{eq:Hermite}
\end{equation}
where $H_\alpha$ is the corresponding product Hermite polynomial. Since a Gaussian has finite moments of every order, conditional Jensen implies finite moments for all normalized derivatives $\p^\alpha f_t/f_t$. Derivatives of $u=\log f_t$ are polynomial combinations of these normalized derivatives, so the weighted polynomial moments entering $J_i$ are finite.

Together with \eqref{eq:Hpositive}, this gives the required integrability of $J_i/H^2$. Let $\chi_R$ be a smooth cutoff equal to one for $|x|\leq R$, zero for $|x|\geq2R$, and satisfying $|\p_i\chi_R|\leq C/R$. Then
\begin{equation}
 \int d^2x\,\chi_R f_t\D_i\!\left(\frac{J_i}{H^2}\right)
 =-\int d^2x\,f_t\frac{J_i}{H^2}\p_i\chi_R
 \longrightarrow0
 \label{eq:cutoff}
\end{equation}
as $R\to\infty$. Hence
\begin{equation}
 \boxed{\avg{\D_i(J_i/H^2)}=0.}
 \label{eq:divzero}
\end{equation}
Taking the expectation of \eqref{eq:divided} gives
\begin{equation}
 D-2c_0\tau Q+c_0\tau^2I\geq0.
 \label{eq:tauineq}
\end{equation}
Choosing $\tau=Q/I$ therefore proves the main theorem.

\begin{theorem}[Certified lower bound]
For every probability measure $\mu$ on $\R^2$ and every $t>0$, with $f_t=\mu*\gamma_t$,
\begin{equation}
 \boxed{I(t)I''(t)\geq\frac{231691}{250000}[I'(t)]^2.}
 \label{eq:lowerthm}
\end{equation}
No symmetry, log-concavity, finite-variance, or finite-entropy assumption on $\mu$ is required.
\end{theorem}

\begin{proof}
Equation \eqref{eq:divided}, positivity of $H$, and positivity of the Gram matrices give \eqref{eq:tauineq}. The positive-time Gaussian estimates justify the cutoff limit \eqref{eq:divzero}. The quadratic in $\tau$ is minimized at $\tau=Q/I>0$, yielding $D\geq c_0Q^2/I$. Finally use $Q=-I'$ and $D=I''$.
\end{proof}

\section{The explicit upper construction}
\label{sec:upper}

The lower certificate proves that no admissible heat flow can have $II''/(I')^2$ below $0.926764$. To obtain an upper bound on the infimum $c_*$, it is enough to construct one admissible density for which the same ratio is smaller than a specified number.

\subsection{A positive Fourier profile}

We use the fixed rational trigonometric polynomial stored in \texttt{third\_upper\_density.json},
\begin{equation}
 g(x)=1+\sum_{(m,n)\in\mathcal M}\frac{A_{mn}}{10^{15}}
 \cos(k_{mn}\cdot x),
 \label{eq:g}
\end{equation}
with triangular-lattice wavevectors
\begin{equation}
 k_{mn}=\left(m-\frac n2,\frac{\sqrt3}{2}n\right),
 \qquad m^2+n^2-mn\leq400.
 \label{eq:k}
\end{equation}
The set contains $729$ modes. Since
\begin{equation}
 \sum_{(m,n)\in\mathcal M}\left|\frac{A_{mn}}{10^{15}}\right|
 =0.468106047404556,
 \label{eq:l1}
\end{equation}
and $|\cos\theta|\leq1$,
\begin{equation}
 g(x)\geq0.531893952595444>0.
 \label{eq:gpositive}
\end{equation}
Thus $\log g$ and all of its derivatives are nonsingular.

A purely periodic function is not normalizable on $\R^2$, so we introduce a broad Gaussian envelope and define
\begin{equation}
 f_R(x)=Z_R^{-1}g(x)^6\exp\!\left(-\frac{|x|^2}{2R^2}\right),
 \qquad R=10^8.
 \label{eq:fR}
\end{equation}
This is a smooth, positive, normalizable probability density. The large value of $R$ makes the envelope vary extremely slowly compared with the periodic structure.

\subsection{Periodic derivative algebra}

First ignore the Gaussian envelope and write
\begin{equation}
 u=6\log g.
 \label{eq:uper}
\end{equation}
Define
\begin{equation}
 B_{ij}\equiv gg_{ij}-g_i g_j,
 \label{eq:Bij}
\end{equation}
so that
\begin{equation}
 u_{ij}=\frac{6B_{ij}}{g^2}.
 \label{eq:uijB}
\end{equation}
A further derivative gives
\begin{equation}
 C_{ijk}\equiv
 g^2g_{ijk}
 -g(g_{ij}g_k+g_{ik}g_j+g_{jk}g_i)
 +2g_i g_j g_k,
 \label{eq:Cijk}
\end{equation}
with
\begin{equation}
 u_{ijk}=\frac{6C_{ijk}}{g^3}.
 \label{eq:uijkC}
\end{equation}

For the periodic density proportional to $g^6$, the common normalization cancels from the ratio $ID/Q^2$. It is therefore convenient to define unnormalized polynomial integrands
\begin{equation}
 P_I\equiv g^6(-h),
 \qquad
 P_Q\equiv g^6q,
 \qquad
 P_D\equiv g^6d_3.
 \label{eq:Pdefs}
\end{equation}
Using \eqref{eq:uijB}--\eqref{eq:uijkC}, these become
\begin{equation}
 \boxed{P_I=-6g^4 B_{ii},}
 \label{eq:PI}
\end{equation}
\begin{equation}
 \boxed{P_Q=36g^2 B_{ij}B_{ij},}
 \label{eq:PQ}
\end{equation}
and
\begin{equation}
 \boxed{P_D=36\left(C_{ijk}C_{ijk}-12B_{ij}B_{jk}B_{ki}\right).}
 \label{eq:PD}
\end{equation}
The factor $g^6$ is useful precisely because it cancels all denominators generated by derivatives of $\log g$. Since $g$ is a finite Fourier polynomial, $P_I$, $P_Q$, and $P_D$ are finite Fourier polynomials as well.

If $\langle\cdot\rangle_{\rm torus}$ denotes the average over one periodic cell, then for the normalized periodic density
\begin{equation}
 I=\frac{\langle P_I\rangle_{\rm torus}}{\langle g^6\rangle_{\rm torus}},
 \quad
 Q=\frac{\langle P_Q\rangle_{\rm torus}}{\langle g^6\rangle_{\rm torus}},
 \quad
 D=\frac{\langle P_D\rangle_{\rm torus}}{\langle g^6\rangle_{\rm torus}}.
 \label{eq:IQDfromP}
\end{equation}
Hence the common normalization cancels,
\begin{equation}
 \boxed{
 \frac{ID}{Q^2}
 =\frac{\langle P_I\rangle_{\rm torus}\langle P_D\rangle_{\rm torus}}
 {\langle P_Q\rangle_{\rm torus}^2}.}
 \label{eq:RperP}
\end{equation}

A finite Fourier polynomial can be written as $P(x)=\sum_k c_ke^{ik\cdot x}$. Averaging over one cell removes every nonzero mode, so only the zero Fourier coefficient survives. The three moments in \eqref{eq:RperP} can therefore be evaluated exactly by coefficient extraction. The resulting periodic ratio is enclosed by
\begin{equation}
 0.94176665012061487
 <\frac{\langle P_I\rangle_{\rm torus}\langle P_D\rangle_{\rm torus}}
 {\langle P_Q\rangle_{\rm torus}^2}
 <0.94176665012061488.
 \label{eq:Rper}
\end{equation}

\subsection{Gaussian envelope and positive heat time}

For the full density \eqref{eq:fR},
\begin{equation}
 \log f_R=6\log g-\frac{|x|^2}{2R^2}-\log Z_R.
 \label{eq:logfR}
\end{equation}
Writing $\eta\equiv R^{-2}$, the Gaussian term shifts the Hessian according to
\begin{equation}
 u_{ij}\longrightarrow u_{ij}-\eta\delta_{ij}.
 \label{eq:Hshift}
\end{equation}
The corresponding unnormalized integrands are
\begin{equation}
 P_{I,R}=P_I+2\eta g^6,
 \label{eq:PIR}
\end{equation}
\begin{equation}
 P_{Q,R}=P_Q+2\eta P_I+2\eta^2g^6,
 \label{eq:PQR}
\end{equation}
and
\begin{equation}
 P_{D,R}=P_D+6\eta P_Q+6\eta^2P_I+4\eta^3g^6.
 \label{eq:PDR}
\end{equation}
Every nonzero triangular-lattice mode has $|k|^2\geq1$, and averaging against the centered Gaussian envelope multiplies it by at most $e^{-R^2|k|^2/2}\leq e^{-R^2/2}$. With $R=10^8$ this is far below $10^{-80}$. Exact rational interval arithmetic therefore gives
\begin{equation}
 0.94176665012061537
 <\frac{I_RD_R}{Q_R^2}
 <0.94176665012061538.
 \label{eq:Rfull}
\end{equation}

Finally define the positive-time heat flow
\begin{equation}
 f_{R,t}=f_R*\gamma_t
 \label{eq:fRt}
\end{equation}
and
\begin{equation}
 \mathcal R_R(t)=\frac{I_R(t)I_R''(t)}{[I_R'(t)]^2}.
 \label{eq:RRt}
\end{equation}
Strict positivity of $g$, together with the Gaussian envelope, implies sufficient domination to make $I_R(t)$, $I_R'(t)$, and $I_R''(t)$ continuous at $t=0$, with $I_R'(0)\neq0$. Therefore
\begin{equation}
 c_*\leq\liminf_{t\downarrow0}\mathcal R_R(t)
 =\mathcal R_R(0)<0.94176665012062.
 \label{eq:upperfinal}
\end{equation}
Combining \eqref{eq:lowerthm} and \eqref{eq:upperfinal} gives \eqref{eq:maininterval}.

\section{Exact verification and reproducibility}
\label{sec:verification}

The lower certificate is stored in
\begin{center}
\texttt{tensor\_lower\_6\_1\_0\_1\_4\_align\_3\_2\_certificate.json}.
\end{center}
The independent verifier \texttt{verify\_aligned\_tensor\_ldl.py} uses rational arithmetic from the Python standard library. It checks
\begin{enumerate}
\item the denominator formula and square decomposition \eqref{eq:Hsquare};
\item the symmetry and dimensions of all eleven Gram matrices;
\item all $491$ positive leading principal minors;
\item every coefficient of the polynomial identity \eqref{eq:certificate}, with zero residual;
\item the two $987$-monomial weighted-divergence components.
\end{enumerate}
The recorded report has status ``exact full identity and positivity passed.'' The upper construction is generated and checked independently by separate scripts.

The semidefinite-programming (SDP) stage is used only to discover a candidate certificate. The chosen ansatz gives a feasibility problem with $925$ exact coefficient equations and $16{,}045$ independent upper-triangular Gram entries before constraints. The floating-point candidate is projected onto the exact affine equations, selected free coordinates are rationalized, and the remaining linear system is solved exactly over the rationals.

For $c=0.926764$, the accepted projected candidate has minimum Gram-block eigenvalue approximately $2.34\times10^{-8}$ and SDP margin approximately $3.54\times10^{-8}$. At $c=0.9267645$ the numerical margin is about $-7.36\times10^{-9}$, while the $c=0.926765$ run returned ``unknown'' and did not produce a strictly feasible interior point. This indicates that the particular aligned degree-six ansatz is numerically close to its feasibility boundary near $0.926764$, but it is not a rigorous infeasibility result and does not constrain other ansatz families.

The principal reproduction commands are
\begin{verbatim}
python verify_aligned_tensor_ldl.py \\
  tensor_lower_6_1_0_1_4_align_3_2_certificate.json
python certify_third_upper.py
python verify_third_upper.py
\end{verbatim}
The full code contains the exact input data, rational certificate, projected interior point, upper-density data, verification scripts, and computational-environment manifest.

\section{Power convexity and open directions}

For $0<\alpha<1$,
\begin{equation}
 (I^\alpha)''
 =\alpha I^{\alpha-2}
 \left[II''-(1-\alpha)(I')^2\right].
 \label{eq:power}
\end{equation}
Since the prefactor is positive, $I^\alpha$ is convex exactly when
\begin{equation}
 II''\geq(1-\alpha)(I')^2.
 \label{eq:powercondition}
\end{equation}
The sharp power threshold is therefore
\begin{equation}
 \alpha_*\equiv1-c_*.
 \label{eq:alphastar}
\end{equation}
Our interval for $c_*$ gives
\begin{equation}
 \boxed{0.05823334987938<\alpha_*\leq0.073236.}
 \label{eq:alphainterval}
\end{equation}
Thus $I^\alpha$ is universally convex for every $\alpha\geq0.073236$, while the explicit upper construction rules out every universal exponent $\alpha\leq0.05823334987938$.

Several questions remain open. Is $c_*$ equal to one of the two endpoints in \eqref{eq:maininterval}, or to neither? Can the lower certificate be improved by choosing a different positive denominator $H$, enlarging the divergence ansatz $J_i$, or exploiting the $O(2)$ decomposition of the local derivative tensors more systematically? Is there an extremizing sequence that explains the numerical proximity of the lower and upper constructions? Finally, extending the lower theorem beyond positive-time Gaussian convolutions requires a separate analysis of integrability and boundary terms.

\section*{Acknowledgements}
The certificate search and exact checks were carried out with Python, SymPy, python-flint, NumPy, SciPy, and CVXOPT. The numerical stage was used for discovery only; the theorem rests on rational reconstruction and independent exact verification. No priority claim is made beyond the explicitly cited literature. We have used ChatGPT in preparing this paper.

\end{document}